\documentclass{article}

\usepackage[final]{neurips_2024}

\usepackage{natbib}
\PassOptionsToPackage{numbers, compress}{natbib}
\usepackage[utf8]{inputenc} 
\usepackage[T1]{fontenc}    
\usepackage{hyperref}       
\usepackage{url}            
\usepackage{booktabs}       
\usepackage{amsfonts}       
\usepackage{nicefrac}       
\usepackage{microtype}      
\usepackage{xcolor}         
\usepackage{bbm}
\usepackage{diagbox}
\usepackage{multirow}

\usepackage{amsmath,amsthm,amssymb}
\usepackage{graphicx}
\usepackage{bm}
\usepackage[capitalise]{cleveref}
\usepackage{enumitem}  
\usepackage{booktabs}
\usepackage{multirow}
\usepackage{wrapfig}
\usepackage{nicefrac}
\usepackage{xcolor}
\usepackage{tabularray}
\usepackage{subcaption,booktabs}

\graphicspath{{figs/}}

\setlist[itemize]{leftmargin=*}
\setlist[enumerate]{leftmargin=*}

\hypersetup{
  colorlinks,
  backref=page,
  linkcolor=magenta,
  urlcolor=violet,
  citecolor=blue!60!green!90
}

\newcommand{\PP}{\mathbb{P}}
\newcommand{\RR}{\mathbb{R}}

\newcommand{\vv}{\bm{v}}

\newcommand{\vx}{\bm{x}}

\newcommand{\vdelta}{\bm{\delta}}

\newcommand{\sD}{\mathcal{D}}

\newcommand{\set}[1]{\{ #1 \}}
\newcommand{\norm}[1]{\lVert #1 \rVert}

\newcommand{\lb}[1]{{#1}^{\perp}}
\newcommand{\ub}[1]{{#1}^{\top}}

\newtheorem{theorem}{Theorem}

\newtheorem{proposition}{Proposition}

\theoremstyle{example}

\theoremstyle{remark}

\theoremstyle{assumption}

\theoremstyle{definition}

\usepackage{arydshln}

\title{Verifiably Robust Conformal Prediction}

\author{%
  Linus Jeary\thanks{Authors contributed equally.} \\
  Department of Informatics \\
  King's College London, UK \\
  \texttt{linus.jeary@kcl.ac.uk} \\
  \And
  Tom Kuipers$^*$ \\
  Department of Informatics \\
  King's College London, UK \\
  \texttt{tom.kuipers@kcl.ac.uk} \\
  \AND
  Mehran Hosseini \\
  Department of Informatics \\
  King's College London, UK \\
  \texttt{mehran.hosseini@kcl.ac.uk} \\
  \And
  Nicola Paoletti\\
  Department of Informatics\\
  King's College London, UK\\
  \texttt{nicola.paoletti@kcl.ac.uk}\\
}

\begin{document}

\maketitle

\begin{abstract}
  Conformal Prediction (CP) is a popular uncertainty quantification
method that provides distribution-free, statistically valid prediction
sets, assuming that training and test data are exchangeable. In such a
case, CP's prediction sets are guaranteed to cover the (unknown) true
test output with a user-specified probability. Nevertheless, this
guarantee is violated when the data is subjected to adversarial
attacks, which often result in a significant loss of coverage.
Recently, several approaches have been put forward to recover CP
guarantees in this setting. These approaches leverage variations of
randomised smoothing to produce conservative sets which account for
the effect of the adversarial perturbations. They are, however,
limited in that they only support $\ell_2$-bounded perturbations and
classification tasks.  This paper introduces \emph{VRCP (Verifiably
Robust Conformal Prediction)}, a new framework that leverages recent
neural network verification methods to recover coverage guarantees
under adversarial attacks. Our VRCP method is the first to support
perturbations bounded by arbitrary norms including $\ell_1$, $\ell_2$,
and $\ell_\infty$, as well as regression tasks. We evaluate and
compare our approach on image classification tasks (CIFAR10, CIFAR100,
and TinyImageNet) and regression tasks for deep reinforcement learning
environments. In every case, VRCP achieves above nominal coverage and
yields significantly more efficient and informative prediction regions
than the SotA.

\end{abstract}

\section{Introduction}
\label{sec:introduction}
Conformal Prediction (CP)
\citep{vovk2005algorithmic,angelopoulos2021gentle} is a popular
uncertainty quantification method. In essence, it is a model-agnostic,
distribution-free framework that allows one to construct prediction
sets that are guaranteed to include the true (unknown) output with
probability greater than \(1 - \alpha\), where \(\alpha \in (0,1)\) is
a user-specified miscoverage/error rate. In other words, for a test
point \((\vx_{n+1}, y_{n+1})\), CP seeks to construct a prediction set
\({C}(\vx_{n+1})\) such that the following coverage (a.k.a. validity)
guarantee holds:
\begin{equation}
  \label{eq:marginal_cp}
  \PP_{\vx_{n+1},y_{n+1}} [ y_{n+1} \in {C}(\vx_{n+1}) ] \geq 1 - \alpha.
\end{equation}
Importantly, the above guarantee holds when the calibration data, used
to construct \({C}(\vx_{n+1})\), and the test point are exchangeable
(a special case is when calibration and test data are i.i.d.).

When exchangeability is violated, e.g., in the presence of test-time
distribution shifts, CP's coverage guarantee~\eqref{eq:marginal_cp}
ceases to hold, and we cannot rely on the prediction sets it produces.
In this work, we address shifts induced by adversarial perturbations
on the test inputs. In particular, we focus on perturbations in the
form of additive \(\ell_p\)-bounded noise.

To recover guarantees under adversarial inputs, the general mechanism
is to inflate the prediction set to permit larger degrees of
uncertainty.  However, special care must be taken to avoid producing
overly large or even trivial sets -- i.e. those containing all
possible outputs -- as such sets do not provide any useful inference.

\begin{figure}[t]
  \centering
  \includegraphics[width=1.01\textwidth]{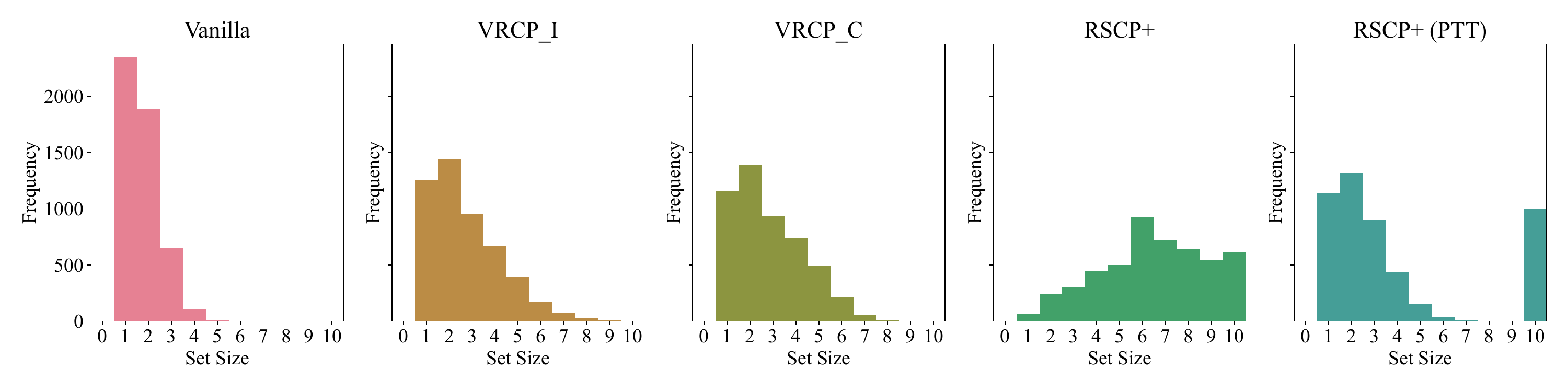}
  \caption{Distribution of prediction set sizes for vanilla conformal
    prediction (vanilla CP) which violates \cref{eq: MarginalVCP}, as
    well as for our proposed robust algorithms (VRCP--I and VRCP--C)
    along with the SotA (RSCP+ and RSCP+ (PTT), see
    \cref{sec:rel-work}) on the CIFAR10 dataset. As we observe,
    VRCP--I and VRCP--C closely resemble the spread of vanilla CP
    prediction set sizes, whilst the SotA falls short of achieving
    this. Here we use an adversarial perturbation of radius
    \(\epsilon = 0.02\), error rate \(\alpha = 0.1\), number of splits
    \(n_{\text{splits}} = 50\) and smoothing parameter (used in RSCP+
    and RSCP+ (PTT)) \(\sigma = 2\epsilon\).}
  \label{figure: Set Size}
  \vspace{-1em}
\end{figure}
\paragraph{Contributions}
We propose a CP framework that provides statistically valid prediction
sets despite the presence of \(\ell_{p}\)-bounded adversarial
perturbations at inference time.  Formally, for any adversarially
perturbed test point \(\tilde{\vx}_{n+1}=\vx_{n+1}+\vdelta\), our
method produces adversarially robust sets \(C_{\epsilon}\) that enjoy
the following guarantee:
\begin{equation}
  \label{eq: MarginalVCP}
  \PP_{} [ y_{n+1} \in C_{\epsilon}(\tilde{\vx}_{n+1}) ] \geq 1 - \alpha \quad \forall \vdelta \ \text{ s.t. } \ \norm{\vdelta}_p\leq \epsilon.
\end{equation}
While CP uses an underlying predictor \(f\), often a neural network
(NN), to construct prediction regions, the novelty of our approach is
to leverage NN verification algorithms to compute upper and lower
output bounds of \(f(\vx')\) for any
\(\vx'\in B_{\epsilon}(\vx) = \{\vx' : \norm{\vx'-\vx}_p\leq
\epsilon\}\). We use these bounds to inflate the CP regions, resulting
in provably robust and \emph{efficient} prediction sets.  To the best
of our knowledge, this is the first work that combines NN verification
algorithms and CP to construct adversarially robust prediction
sets. We call our method \emph{VRCP (Verifiably Robust Conformal
  Prediction)}.

Recent work (discussed in Section~\ref{sec:rel-work}) achieves
adversarially robust coverage using probabilistic methods,
specifically, randomised smoothing~\citep{cohen2019certified}. Our
approach overcomes some of the theoretical and empirical drawbacks of
these prior methods, which are restricted to classification tasks with
\(\ell_2\)-norm bounded guarantees and are overly conservative in
practice.

Thanks to our verification-based approach, VRCP is the first to extend
adversarially robust conformal prediction to regression tasks and the
first to go beyond \(\ell_2\)-norm bounded guarantees. In
\cref{sec:theory}, we introduce two versions of VRCP that apply
verification at calibration and inference time, respectively. Further,
in \cref{sec:evaluation}, we empirically validate our theoretical
guarantees and demonstrate a direct improvement over previous work in
terms of prediction set efficiency (i.e., average set size) compared
to prior work. \cref{figure: Set Size} shows an extract of our
results on CIFAR10, demonstrating that VRCP yields more informative
(tighter) prediction regions, a trend that we observe experimentally
across all our benchmarks.


%
\section{Preliminaries}
\label{sec:preliminaries}

We denote with \(\RR_+\) the set of positive real numbers. Vectors
\(\vx \in \RR^d\) are shown in bold italic and scalars \(x \in \RR\)
in italic typeface.  We denote the norm used to make \(\RR^d\) a
normed vector space by \(\norm{\cdot}\). This could for instance be
\(\ell_1, \ell_2\), or \(\ell_{\infty}\)-norm. Whenever a specific
norm is intended, we indicate it using an index, e.g.,
\(\norm{\cdot}_2\) indicates the \(\ell_2\)-norm. We denote the
\(\epsilon\)-ball around a point \(\vx \in \RR^d\) with respect to the
used norm by \(B_{\epsilon}(\vx)\).

\subsection{Conformal Prediction}
We provide a brief overview of the inductive (or split) vanilla CP
approach. Suppose we have a dataset \(\sD\) containing pairs
\((\vx,y)\) sampled i.i.d.\ from an (unknown) data-generating
distribution over a feature space \(X \subseteq \mathbb{R}^d\) and
label space \(Y\) such that
\(\sD = \set{(\vx_1, y_1), \dots, (\vx_m, y_m)}\).

We partition the dataset into disjoint training and calibration sets
\(\sD_{\text{train}}\) and \(\sD_{\text{cal}}\), letting
\(n = \vert \sD_{\text{cal}} \vert\).  We fit a predictor \(f\) on
\(\sD_{\text{train}}\) and define a score function
\(S: (X \times Y) \to \RR\) as some notion of prediction error, such
as \(S(\vx,y)= \| f(\vx) - y \|\) when \(f\) is a regressor, or
\(S(\vx,y)= 1 - f(\vx)_y\) when \(f\) is a classifier with
\(f(\cdot)_y\) being \(y\)'s predicted likelihood.

After applying the score function to all calibration points, we
construct the score distribution as
\(F = \nicefrac{\delta_{\infty}}{(n+1)} + \sum_{i=1}^{n}
\nicefrac{\delta_{s_i}}{n+1}\), where \(\delta_s\) is the Dirac
distribution with parameter \(s\), \(s_i=S(\vx_i,y_i)\) and
\(\delta_{\infty}\) represents the unknown score (potentially
infinite) of the test point.

Given a miscoverage/error rate \(\alpha\) and a test point
\((\vx_{n+1}, y_{n+1})\), we define the prediction set
\(C(\vx_{n+1})\) by including all labels that appear sufficiently
likely w.r.t.\ the score distribution:
\(C(\vx_{n+1}) = \set{y\in Y : S(\vx_{n+1}, y) \leq Q_{1-\alpha}(F)
}\), where \(Q_{1-\alpha}(F)\) is the \(1-\alpha\) quantile of
\(F\). This set satisfies the marginal coverage guarantee in
\cref{eq:marginal_cp} if the test point and the calibration points are
exchangeable.

\subsection{Adversarial Attacks}
Neural networks have been shown to be vulnerable to \emph{adversarial
  attacks}, i.e., small changes to their input that jeopardise the
prediction \citep{Szegedy+14adversarial,biggio2018wild}. This notion
can be formally defined as maximising an adversarial objective
function (e.g., the loss of the true label) subject to
\(\norm{\vx - \tilde{\vx}} \leq \epsilon\). Alternatively, it can be
defined as finding an adversarial example \(\tilde{\vx} \in \RR^m\),
such that \(\norm{\vx - \tilde{\vx}} \leq \epsilon\) and
\(\norm{f(\tilde{\vx}) - y} \geq \delta\) for a given neural network
\(f: \RR^d \to \RR^n\).

\subsection{Neural Network Verification}
Various approaches have been proposed to verify the robustness of NNs
against adversarial attacks. These approaches can be divided into
complete and incomplete algorithms. Given a neural network \(f\), a
verifier is \emph{complete} if it allows computing exact bounds
\(f^{\perp}\) and \(f^{\top}\) for the image
\(f(B_{\epsilon}(\vx)) = \set{f(\vx') : \vx' \in B_{\epsilon}(\vx)}\),
i.e., such that
\begin{equation}
  \label{eq: Exact Bounds}
  \lb{f} = \min_{\vx' \in B_{\epsilon}(\vx)} \set{f(\vx')},
  \quad
  \ub{f} = \max_{\vx' \in B_{\epsilon}(\vx)} \set{f(\vx')},
\end{equation}
where \(\min\) and \(\max\) are computed coordinate-wise for
vector-valued NNs. A verifier is \emph{incomplete}, but
\textit{sound}, if it computes bounds that are valid but not exact,
i.e., such that:
\begin{equation}
  \label{eq: Approximate Bounds}
  \lb{f} \leq \min_{\vx' \in B_{\epsilon}(\vx)} \set{f(\vx')},
  \quad
  \ub{f} \geq \max_{\vx' \in B_{\epsilon}(\vx)} \set{f(\vx')}.
\end{equation}
Our results are verifier-agnostic, meaning that they are valid for any
verifier that can produce exact bounds (as in \cref{eq: Exact Bounds})
or conservative bounds (as in \cref{eq: Approximate Bounds}),
depending on the completeness or incompleteness of the verifier used.
The fastest and simplest way to compute the bounds in \cref{eq:
  Approximate Bounds} is to propagate the bounds on the input
\(B_{\epsilon}(\vx)\) through the network to compute the output
bounds. Several methods based on this approach have been proposed
\citep{Gowal+18,Wang+19,Zhang+18crown,BattenHL24,Lopez+23}. At the
expense of fast computation speed, these methods may result in loose
bounds in \cref{eq: Approximate Bounds}. On the other hand, several
complete methods \citep{PulinaT10,Katz+17,HosseiniL23} for NN
verification have been put forward. Even though these methods compute
exact bounds, their downside is their high computational cost.


%
\section{Related Work}
\label{sec:rel-work}
\paragraph{Adversarially Robust Conformal Prediction}
\cite{gendler2021adversarially} introduced an algorithm called
Randomly Smoothed Conformal Prediction (RSCP) that integrates
randomised smoothing
\citep{duchi2012randomized,cohen2019certified,salman2019provably} with
CP to provide robust coverage under adversarial attacks.  RSCP
replaces the CP score function \(S(\vx,y)\) with a smoothed score
\(\widetilde{S}(\vx,y)\) obtained by averaging the values of
\(S(\vx + \vv,y)\) over \(n_{\text{MC}}\) realisations of a Gaussian
noise vector \(\vv\sim \mathcal{N}(0, \sigma^2 I)\), for a given
smoothing level \(\sigma\).  To correct for potential \(\ell_2\)-norm
bounded \(\epsilon\) perturbations at inference time, the critical
value computed over the smoothed scores is inflated by
\(\epsilon/\sigma\).  This method produces empirically sound results,
but the provided formal guarantees were found to be invalid in a later
work \citep{yan2023provably}, as discussed below.

\paragraph{Provably Robust Conformal Prediction}
\cite{yan2023provably} address the issue with the robustness guarantee
of \cite{gendler2021adversarially} by correctly bounding the
estimation error caused by the Monte-Carlo sampling used when
generating the smoothed scores. The bound introduces an additional
hyperparameter \(\beta\) such that they now find the
\(Q_{1-\alpha+2\beta}\) of smooth calibrated scores and inflate by
Hoeffding's bound \(\sqrt{\nicefrac{-\ln(\beta)}{2n_{\text{MC}}}}\)
before correcting by \(\epsilon/\sigma\). Furthermore, the smoothed
scores of the test points are decreased by an empirical Bernstein
bound. This further inflation of the critical value and deflation of
smooth scores for each test point often cause their amended algorithm,
so-called RSCP+, to generate trivial prediction sets.

To address this issue, the authors introduce two methods to improve
the efficiency of RSCP+.  Firstly, they use robustly calibrated
training (RCT), a training-time regularisation technique that
penalises NN parameters that contribute to high scores for the true
label. Our approach assumes that the underlying classifier is given;
hence, we do not evaluate RCT in our experiments.

Secondly, they implement a post-training transformation (PTT), which
aims to decrease the values of the smoothed calibration scores that
lie between \(Q_{1-\alpha}\) (the critical value of the base scores)
and \(\widetilde{Q}_{1-\alpha}\) (that of the smoothed scores). To
this purpose, they transform the CDF of the smoothed scores
\(\widetilde{S}\) by composing learned ranking and sigmoid
transformations with hyperparameters \(b\) and \(T\) using a holdout
set \(\sD_{\text{hold}}\) sampled i.i.d from \(\sD_{\text{cal}}\).
PTT however is not theoretically guaranteed to improve the average set
sizes computed by RSCP+ and, in many cases, its efficacy is largely
dependent on how representative the sampled holdout set is of athe
calibration set. We demonstrate the effect of PTT's sample dependence
empirically in \cref{subsec: classExp}.

\paragraph{Probabilistically Robust Conformal Prediction}
\cite{ghosh2023probabilistically} also focus on the adversarial
setting but maintain a relaxed form of robust coverage, where input
perturbations \(\vdelta\) are drawn from a specific distribution and
only a proportion of such perturbations are sought to be covered. In
contrast, we do not make assumptions about the noise distribution, and
we account for any \(\epsilon\)-bounded perturbation.

All the works\footnote{We are aware of related contemporaneous work by
  \citet{ZargarbashiAB24}. However, at the time of writing, neither
  the manuscript nor the code were available.} discussed here rely on
randomised smoothing \citet{duchi2012randomized} and as such are
limited to the \(\ell_2\)-norm. In contrast, our VRCP approach relies
on NN verifiers, can be used with any \(\ell_p\)-norms supported by
the verification method, and does not require smoothing
hyperparameters or holdout sets.


%
\section{Verifiably Robust Conformal Prediction (VRCP)}
\label{sec:theory}
In this section, we formally introduce two variants of VRCP. Both
methods allow us to construct adversarially robust prediction sets at
inference time.

The first variant, \textit{VRCP via Robust Inference (VRCP--I)},
employs NN verification at inference time to compute a lower bound of
the score for the given test input (best-case score), thereby
obtaining more conservative regions.  The calibration procedure is
computed as in standard CP.  The second variant, \textit{VRCP via
  Robust Calibration (VRCP--C)}, instead uses NN verification at
calibration time to derive upper bounds for the calibration scores
(worst-case), thereby obtaining a more conservative calibration
threshold (critical value). This allows us to use the regular scores
at inference time, without requiring NN verification.

\subsection{Verifiably Robust Conformal Prediction via Robust Inference (VRCP--I)}

Given a calibration set \(\sD_{\text{cal}}\), a test input
\(\vx_{n+1}\), and score function \(S(\cdot, \cdot)\), we compute the
prediction set for \(\vx_{n+1}\) as follows.
\begin{enumerate}
\item For each \(y \in Y\) we compute,
  \begin{equation}
    \label{eq: VRCP_I_score}
    \lb{s}(\vx_{n+1}, y) \leq \inf_{\vx' \in B_{\epsilon}(\vx_{n+1})} S(\vx', y).
  \end{equation}  
\item The robust prediction set is then defined as
  \begin{equation}
    \label{eq: VRCP}
    C_{\epsilon}(\vx_{n+1}) = \set{y : \lb{s}(\vx_{n+1}, y) \leq Q_{1 - \alpha}(F)}.
  \end{equation}
\end{enumerate}

Below, we show that we are able to maintain the marginal coverage
guarantee from \cref{eq: MarginalVCP} for any \(\ell_p\)-norm bounded
adversarial attack.

\begin{theorem}
  Let \(\tilde{\vx}_{n+1}=\vx_{n+1}+\vdelta\) for a clean test sample
  \(\vx_{n+1}\) and \(\norm{\vdelta}_p\leq \epsilon\). The prediction
  set \(C_{\epsilon}(\tilde{\vx}_{n+1})\) defined in \cref{eq: VRCP}
  satisfies
  \( \PP \left[ y_{n+1} \in C_{\epsilon}(\tilde{\vx}_{n+1}) \right]
  \geq 1 - \alpha.  \)
\end{theorem}
\begin{proof}
  We obtain that
  \begin{align*}
    \PP \left[ y_{n+1} \in C_{\epsilon}(\tilde{\vx}_{n+1}) \right]& = \PP \left[\lb{s}(\tilde{\vx}_{n+1},  y_{n+1}) \leq Q_{1 - \alpha}(F) \right]\\
                                                             & \geq \PP \left[ \inf_{\vx' \in B_{\epsilon}(\tilde{\vx}_{n+1})}S(\vx',  y_{n+1}) \leq Q_{1 - \alpha}(F) \right] & \text{by \cref{eq: VRCP_I_score}}\\
                                                             & \geq \PP \left[S(\vx_{n+1},  y_{n+1}) \leq Q_{1 - \alpha}(F) \right] \ \geq \ 1 - \alpha. \qedhere
  \end{align*}
\end{proof}
\subsection{Verifiably Robust Conformal Prediction via Robust Calibration}
Given a calibration set \(\sD_{\text{cal}}\), a test input
\(\vx_{n+1}\), and score function \(S(\cdot, \cdot)\), we compute the
robustly calibrated prediction set for \(\vx_{n+1}\) as follows.

\begin{enumerate}
\item We compute the upper-bound calibration distribution as:
  \begin{equation}
    \label{eq: VRCP_C_score}
    \ub{F} = \frac{\delta_{\infty}}{(n+1)} + \sum_{i=1}^{n} \frac{\delta_{\ub{s}_i}}{n+1} \text{, where }
    \ub{s}_{i} \geq \sup_{\vx' \in B_{\epsilon}(\vx_{i})} S(\vx', y_i).
  \end{equation}
\item The robust post-calibration prediction set is then defined as
\end{enumerate}
\begin{equation}
  \label{eq: VRCCP}
  C_{\epsilon}(\vx_{n+1}) = \set{y : S(\vx_{n+1}, y) \leq Q_{1 - \alpha}(\ub{F})}.
\end{equation}

\begin{theorem}
  \label{thm: VRCPCmarg}
  Let \(\tilde{\vx}_{n+1}=\vx_{n+1}+\vdelta\) for a clean test sample
  \(\vx_{n+1}\) and \(\norm{\vdelta}_p\leq \epsilon\). The prediction
  set \(C_{\epsilon}(\tilde{\vx}_{n+1})\) defined in \cref{eq: VRCCP}
  satisfies
  \( \PP \left[ y_{n+1} \in C_{\epsilon}(\tilde{\vx}_{n+1}) \right]
  \geq 1 - \alpha.\)
\end{theorem}
\begin{proof}
  We have that
  \begin{align}
    \PP \left[ y_{n+1} \in C_{\epsilon}(\tilde{\vx}_{n+1}) \right] 
    & = \PP \left[S(\tilde{\vx}_{n+1},  y_{n+1}) \leq Q_{1 - \alpha}\left(\ub{F}\right)\right]\nonumber\\
    & \geq \PP \left[S(\tilde{\vx}_{n+1},  y_{n+1}) \leq
      Q_{1 - \alpha} \left( \left\{\sup_{\vx' \in B_{\epsilon}(\vx_{i})} S(\vx', y_{i})\right\}_{i=1}^{n} \cup \{\infty\} \right) \right]\nonumber\\
    & \geq \PP \left[\sup_{\vx' \in B_{\epsilon}(\vx_{n+1})} \!\!\! S(\vx', y_{n+1}) \leq
      Q_{1 - \alpha} \left( \left\{\sup_{\vx' \in B_{\epsilon}(\vx_{i})} S(\vx', y_{i})\right\}_{i=1}^{n} \!\!\! \cup \{\infty\} \right) \right]\nonumber\\
    & \geq 1 - \alpha\nonumber
  \end{align}
  Let \(P^\top\) denote the distribution of \(({\vx}^\top,y)\) where
  \({\vx}^\top=\text{argsup}_{\vx' \in B_{\epsilon}(\vx)} S(\vx', y)\)
  and \((\vx,y)\sim P\). The final inequality above holds since it is
  equivalent to constructing a CP prediction set using \(n\) i.i.d
  realisations of \(P^\top\) and evaluating it on \({\vx}_{n+1}\)
  \(\sim P^\top\). The resulting set will include the true test output
  \(y_{n+1}\) with probability at least \(1-\alpha\).
\end{proof}
\subsection{Computation of score bounds}
\paragraph{Classification} In the classification setting, we use the
score function proposed in
\citep{lei2013distribution,gendler2021adversarially}:
\begin{equation}
  \label{eq: HPS}
  S(\vx,y)=1-f(\vx)_{y},
\end{equation}
where \(f(\vx)_{y}\in (0,1)\) is the model-predicted likelihood for
label \(y\).  In this setting, to compute \(\lb{s}\) and \(\ub{s}\)
(required by VRCP--I and VRCP--C, respectively), it suffices to use NN
verification algorithms to derive the output bounds for
\(f(\vx)\). Specifically, in VRCP--I, for a test input \(\vx_{n+1}\)
and for each \(y\in Y\) we derive \(\lb{s}(\vx_{n+1}, y)\) as
\begin{equation}
  \label{eq: lbClassScore}
  \lb{s}(\vx_{n+1}, y) = 1-\ub{f(\vx_{n+1})}_y, 
\end{equation}
where \(\ub{f(\vx_{n+1})}_y\) denotes the upper bound computed by the
neural network verifier for the model-predicted likelihood of label
\(y\in Y\) and input \(\vx_{n+1}\).

In VRCP--C, for each calibration point \((\vx_i,y_i)\) we compute
\(\ub{s}(\vx_i, y_i)\) as
\begin{equation}
  \label{eq: ubClassScore}
  \ub{s}(\vx_{i}, y_i) = 1-\lb{f(\vx_i)}_{y_i},
\end{equation}
where \(\lb{f(\vx_{i})}_{y_i}\) denotes the lower bound of the model
output for label \(y_i\) given input \(\vx_i\).

\paragraph{Regression} 
In the regression tasks, we follow the conformalized quantile
regression (CQR) methodology proposed by
\citep{romano2019conformalized}. We train quantile regressors
\(f_{\alpha_{\text{low}}}\) and \(f_{\alpha_{\text{high}}}\) to
estimate the \(\alpha_{\text{low}} = \alpha/2\) and
\(\alpha_{\text{high}} = 1-\alpha/2\) quantiles of \(y \mid \vx\). In
CQR, we use the following score function:
\begin{equation}
  \label{eq: pinballScore}
  S(\vx,y) = \max{\{ f_{\alpha_{\text{low}}}(\vx) - y, y - f_{\alpha_{\text{high}}}(\vx) \} }.
\end{equation}
Unlike classification, where the label space is discrete, we cannot
construct the region explicitly by enumerating all possible outputs
\(y\). Instead, the prediction region for a given test point
\(C(\vx_{n+1})\) is constructed implicitly, by adjusting the quantile
predictions by the critical value of the calibration distribution
\(Q_{1-\alpha}(F)\), as follows:
\begin{equation}
  \label{eq:cqr_region}
  C(\vx_{n+1}) = \left[f_{\alpha_{\text{low}}}(\vx_{n+1}) - Q_{1-\alpha}(F), f_{\alpha_{\text{high}}}(\vx_{n+1}) + Q_{1-\alpha}(F)\right]
\end{equation}

In both VRCP--C and VRCP--I, the score function leverages an NN
verifier to derive the bounds over the upper and lower quantiles of
the model. In VRCP--C, we compute the worst-case calibration scores
as:
\begin{equation}
  \label{eq:pinball_vrcp_c}
  \ub{s}(\vx_i,y_i) = \max{\{ \ub{f}_{\alpha_{\text{low}}}(\vx_i) - y_i, y_i - \lb{f}_{\alpha_{\text{high}}}(\vx_i) \} }.
\end{equation}

In VRCP--I for classification, for each output we check inclusion in
\(C_{\epsilon}\) by using the best-case score \(s^{\bot}\). As
explained above, explicit enumeration is infeasible for regression,
and so we construct our robust region by replacing predicted quantiles
in \cref{eq:cqr_region} with their conservative approximations, as
follows:
\begin{equation}
  \label{eq:robust_cqr_region}
  C_{\epsilon}(\vx_{n+1}) = \left[\lb{f}_{\alpha_{\text{low}}}(\vx_{n+1}) - Q_{1-\alpha}(F), \ub{f}_{\alpha_{\text{high}}}(\vx_{n+1}) + Q_{1-\alpha}(F)\right]
\end{equation}

The above-defined region is equivalent to enumerating all possible
outputs $y$, and for each, considering the best-case score
\(\lb{s}(\vx_{n+1},y) = \max{\{
  \lb{f}_{\alpha_{\text{low}}}(\vx_{n+1}) - y, y -
  \ub{f}_{\alpha_{\text{high}}}(\vx_{n+1}) \} } \). A proof is
available in Appendix~\ref{sec: proofDetails}.

A nice property of both VRCP--I and VRCP--C is that they guarantee
that they can only increase the size of the prediction set for any
input \({\vx}\) compared to vanilla CP, thus will always attain at
least as much coverage as the vanilla CP procedure. Moreover, as we
show in \cref{sec:evaluation}, both algorithms do not trivially
inflate the size of the prediction sets and maintain a similar
distribution of set sizes. This is formalised in the \cref{prop:
  Containment}, which is proved in Appendix \ref{proof:containment}.
\begin{proposition}
  \label{prop: Containment}
  Let \(C(\vx)\) and \(C_{\epsilon}(\vx)\) be the prediction sets
  obtained using vanilla CP and VRCP (using VRCP--I or VRCP--C),
  respectively. For any input \(\vx\), we have that
  \(C(\vx) \subseteq C_{\epsilon}(\vx)\).
\end{proposition}
%


%
\section{Evaluation}
\label{sec:evaluation}
We evaluate VRCP--I and VRCP--C on classification (image) and
regression (RL) benchmarks, and compare them against the SotA
approaches on each benchmark. For all the networks used, we did not
perform adversarial training as we assume that the attack budget
\(\epsilon\) is unknown at training time. Nonetheless, both our
approaches can benefit from adversarial training, as it results in
models that are more verifiable and have tighter bounds for the same
attack budget.\footnote{Code for the experiments is available at:
  \href{https://github.com/ddv-lab/Verifiably_Robust_CP}{https://github.com/ddv-lab/Verifiably\_Robust\_CP}}

\subsection{Classification Experiments} \label{subsec: classExp} We
evaluate each method using a nominal coverage of \(1-\alpha=0.9\) and
report the 95\% confidence intervals for coverage and average set
sizes computed over 50 splits (\(n_\text{splits}=50\)) of the
calibration, holdout and test set.

\paragraph{Bounds}
We use the verification library auto\_LiRPA \citep{xu2020automatic} to
compute the output bounds for \(f(\vx)\) required in \cref{eq:
  lbClassScore} and \cref{eq: ubClassScore} for VRCP--I and VRCP--C
respectively. In particular, we use two SotA GPU-parallelised
incomplete NN verification algorithms, CROWN \cite{zhang2018efficient}
and \(\alpha\)-CROWN \cite{xu2020fast}. In brief, CROWN performs
linear bound propagation and \(\alpha\)-CROWN employs a
branch-and-bound algorithm to tighten the CROWN bounds at the expense
of slower verification times. Therefore, we use CROWN to compute the
output bounds for the TinyImageNet model and \(\alpha\)-CROWN for the
smaller CIFAR10 and CIFAR100 models.

Our CIFAR10 model with \(\alpha\)-CROWN takes \(\approx 0.5\)s per
image to compute bounds with \(\epsilon=0.03\), whereas our larger
CIFAR100 model takes \(\approx 7.2\)s with
\(\epsilon=0.02\). Comparatively, computing the smoothed scores takes
\(\approx 0.09\)s per image to compute on both models under the same
respective $\epsilon$ values. The largest model for the TinyImageNet
dataset uses CROWN to compute bounds at a rate of \(\approx 0.2\)s per
image whereas the smoothed scores take \(\approx 0.24\)s. All
measurements are made with respect to the hardware details listed in
\cref{sec: modelDetails}.

\paragraph{Attacks}
We use the PGD attack algorithm \citep{madry2017towards}, which is a
popular white-box attack algorithm to generate adversarial inputs with
respect to either the \(\ell_2\) or \(\ell_\infty\)-norm.

\paragraph{Models}
For all datasets, we train a CNN model on training set images with
random crop and horizontal flip augmentations. Full model details are
outlined in the appendix.

\paragraph{Hyperparameters}
RSCP+ based approaches use \(\sigma = 2\epsilon\), \(\beta=0.001\) and
those with PTT use \(|\sD_\text{hold}|=500\), \(b = 0.9\) and
\(T = 1/400\). For PGD, we choose a step size of \(1/255\) and compute
\(100\) steps for each attack. For CIFAR10 and CIFAR100
\(|\sD_\text{train}| = 50{,}000\) and for TinyImageNet
\(|\sD_\text{train}| = 100{,}000\). For all datasets
\(|\sD_\text{cal}| = 4{,}500\) and \(|\sD_\text{test}| = 5{,}000\).

\paragraph{Results}
In \cref{table:1}, we benchmark both our methods against the baseline
vanilla CP (which is agnostic of the attack), RSCP+ and RSCP+ with
PTT. At inference time, images are attacked using PGD to generate
\(\ell_2\)-norm bounded attacks with \(\epsilon=0.02\) for CIFAR100
and TinyImageNet, and \(\epsilon=0.03\) for CIFAR10.

In all domains, the vanilla CP method fails to construct valid
prediction sets with nominal marginal coverage, as expected.  RSCP+
maintains robust marginal coverage but produces trivial prediction
sets in all settings due to the highly conservative inflation of the
threshold with respect to the calibration scores. Using PTT improves
RSCP+'s performance but introduces significant variance in the set
sizes: in many cases, PTT still produces trivial prediction sets and
is heavily dependent on the sampled holdout set for RSCP+ to generate
useful predictions.

Both of our methods have minimal sample dependence, as demonstrated by
a very small variability in coverage and size over the \(50\)
splits. We obtain prediction sets with substantially smaller average
sizes than the other robust approaches, and hence, they provide more
informative uncertainty estimates. VRCP--I provides slightly more
efficient regions than VRCP--C. Still, it implies additional
computational overhead at inference time because it requires computing
bounds via NN verification for each test sample. In contrast, in
VRCP--C, bounds are computed only once at calibration time. On the
other hand, in an environment where we may want to change \(\epsilon\)
for different test points at inference time, VRCP--I would be a sound
choice, while VRCP--C would require re-calibration.

\begin{table}
  \centering
  \small 
  \caption{Marginal Coverage and Average Set Sizes for different
    methods on CIFAR10, CIFAR100 and TinyImageNet. All results record
    a 95\% confidence interval with \(n_\text{splits}=50\),
    \(\alpha = 0.1\), \(\sigma = 2\epsilon\),
    \(n_{\text{MC}} = 1024\), \(\epsilon = 0.03\) for CIFAR10 and
    \(\epsilon = 0.02\) otherwise.}
  \label{table:1}
  \resizebox{\columnwidth}{!}{%
    \begin{tabular}{c c c c c c c}
      \toprule
      & \multicolumn{2}{c}{CIFAR10} & \multicolumn{2}{c}{CIFAR100} & \multicolumn{2}{c}{TinyImageNet} \\
      \cmidrule(rl){2-3}
      \cmidrule(rl){4-5}
      \cmidrule(rl){6-7}
      Method & Coverage & Size & Coverage & Size & Coverage & Size \\
      \toprule
      Vanilla & {\color[HTML]{CC0000} 0.878±0.002} & {\color[HTML]{CC0000} 1.721±0.008} & {\color[HTML]{CC0000} 0.890±0.002} & {\color[HTML]{CC0000} 6.702±0.058} & {\color[HTML]{CC0000}0.886±0.002} & {\color[HTML]{CC0000} 38.200±0.252}\\
      \midrule 
      RSCP+ & 1.000±0.000 & 10.000±0.000 & 1.000±0.000 & 100.000±0.000 & 1.000±0.000 & 200.000±0.000 \\
      RSCP+ (PTT) & 0.983±0.008 & 8.357±0.780 & 0.925±0.010 & 26.375±9.675 & 0.931±0.013 & 90.644±20.063 \\
        \midrule 
      VRCP--I & 0.986±0.000 & \textbf{4.451±0.011} & 0.971±0.001 & \textbf{22.530±0.107} & 0.958±0.001 & \textbf{72.486±0.311} \\
      VRCP--C & 0.995±0.000 & 5.021±0.010 & 0.983±0.000 & 23.676±0.131 & 0.965±0.001 & 77.761±0.352 \\
      \bottomrule
    \end{tabular}
  }
\end{table}
\paragraph{Effect of increasing adversarial noise}
\cref{figure:1} shows the impact of increasing \(\epsilon\) across all
evaluated robust methods. Our methods consistently produce smaller
average set sizes with minor sample dependence, and simultaneously
provide a more conservative marginal coverage than RSCP+ (PTT). We
remark that, unlike RSCP+, we do not require a holdout set or any
score function transformations.

\begin{figure}
  \centering
    \begin{subfigure}{\textwidth}
    \hspace{-0.5em}
    \includegraphics[width=0.5\textwidth]{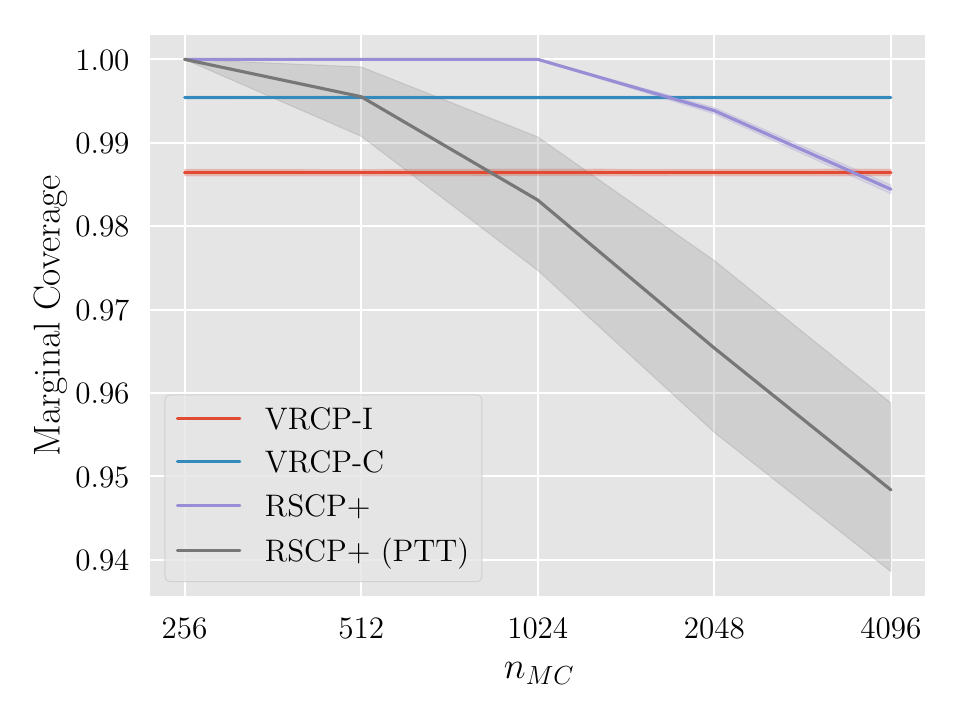}
    \hspace{-0.85em}
    \includegraphics[width=0.5\textwidth]{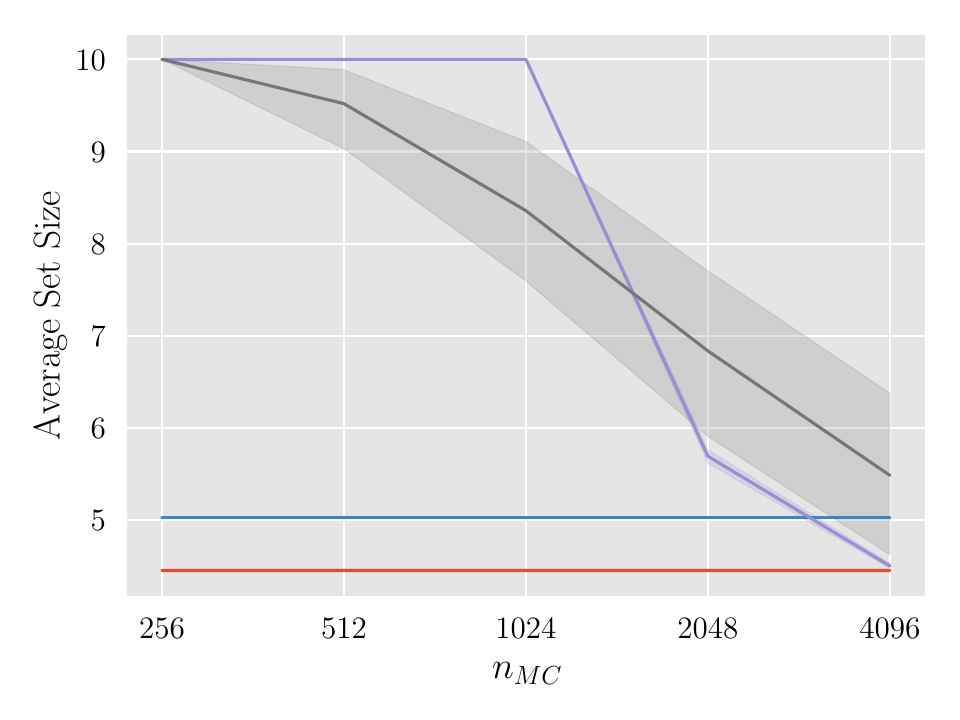}
    \subcaption{Varying \(n_{\text{MC}}\) for \(n_{\text{splits}} = 50\), \(\epsilon = 0.03\),  \(\alpha = 0.1\),  and \(\sigma = 2\epsilon\).}
    \label{figure:2}
  \end{subfigure}
  \hfill
  \begin{subfigure}{\textwidth}
    \hspace{-0.75em}
    \includegraphics[width=0.519\textwidth]{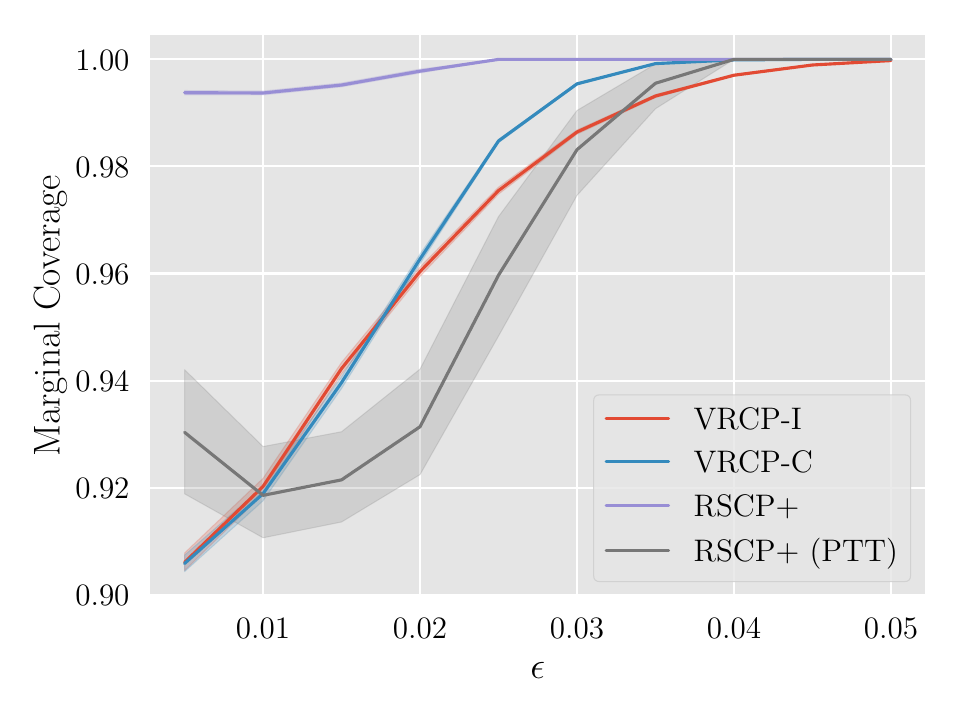}
    \hspace{-0.6em}
    \includegraphics[width=0.519\textwidth]{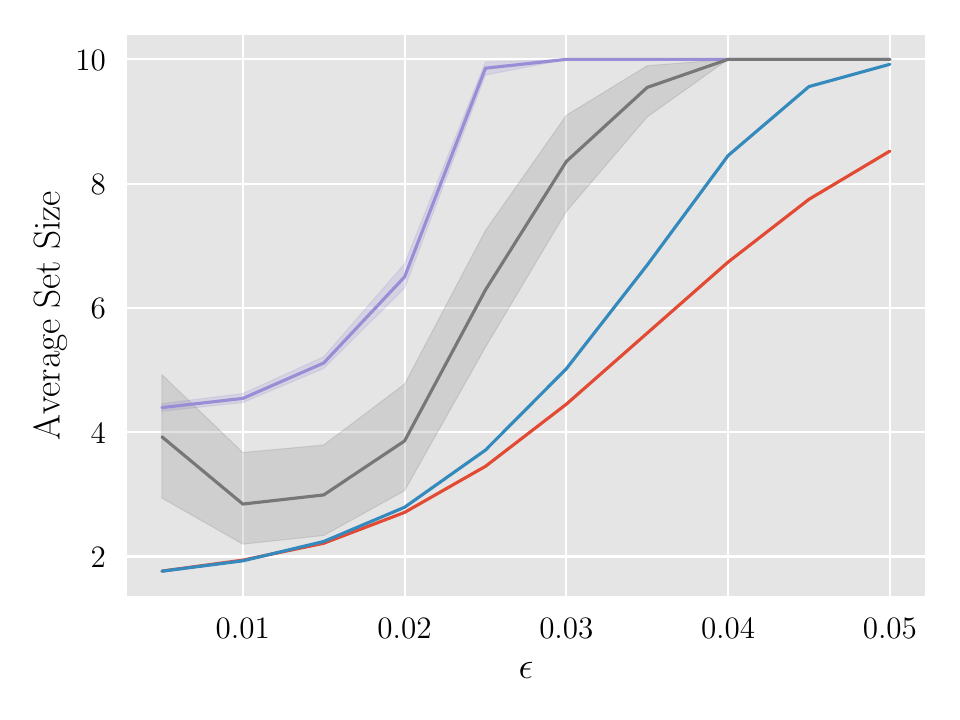}
    \subcaption{Varying the value of \(\epsilon\) for
      \(n_\text{splits}=50\), \(\sigma = 2\epsilon\),
      \(\alpha = 0.1\), and \(n_{\text{MC}}\) = 1024.}
    \label{figure:1}
  \end{subfigure}
  \caption{Marginal Coverage and Average Set Sizes on CIFAR100 with
    95\% confidence intervals.}
  \label{fig: Main Figure}
\end{figure}
\paragraph{Effect of increasing Monte-Carlo samples}
\cref{figure:2} displays the influence of the \(n_\text{MC}\)
hyperparameter on the RSCP+ based methods with respect to our CIFAR10
model. Whilst increasing samples improves the performances of
randomised smoothing approaches, we incur a large computational
overhead when computing the smoothed scores. In our experiments in
\cref{table:1} we fix \(n_\text{MC}=1024\) which is four times larger
than the value for \(n_\text{MC}\) used in previous work
\citep{gendler2021adversarially,yan2023provably} as a trade-off
between prediction quality and computation.

\begin{wraptable}{r}{6cm}
  \caption{Marginal Coverage and Average Set Sizes for \(\epsilon\)
    perturbations with respect to the \(\ell_\infty\)-norm on the
    CIFAR10 dataset. All results record a 95\% confidence interval
    with \(n_{\text{splits}} = 50\), \(\alpha = 0.1\) and
    \(\epsilon = 0.001\).}
  \label{table:2}
  \begin{tabular}{ccc}
    \toprule
    & \multicolumn{2}{c}{CIFAR10} \\
    \cmidrule(lr){2-3}
    Method &  Coverage  & Size \\
    \midrule
    Vanilla         & {\color[HTML]{CC0000} 0.872±0.002} & {\color[HTML]{CC0000} 1.737±0.007} \\
    VRCP--I         & 0.947±0.001 & \textbf{2.262±0.008} \\
    VRCP--C         & 0.931±0.001 & 2.342±0.008 \\
    \bottomrule
  \end{tabular}
\end{wraptable}
\paragraph{Beyond \(\ell_2\)-norm bounded attacks}
\cref{table:2} demonstrates that both of our methods generalise to
other \(\ell_p\)-bounded perturbations other than for when \(p=2\)
which RSCP+ is limited to. In particular, we examine the
\(\ell_\infty\), where even a small \(\epsilon\) can cause
misclassification. We experiment using CIFAR10 and use
\(\epsilon=0.001\). PGD is used to generate \(\ell_\infty\)-bounded
adversarial examples.

\paragraph{Set size distribution}
From \cref{figure: Set Size} we can visually examine the sample
dependency issue that the PTT introduces. In the splits where the
holdout set allows the PTT to make an informative transformation,
RSCP+ is able to make quite reasonable predictions, otherwise, RSCP+
just returns trivial sets. This is clearly an undesirable property and
adds significant variance to the predictions.

Both of our methods increase the spread of the average set sizes to
account for the presence of adversarial examples whilst still
maintaining a consistent distribution.

\subsection{Regression Experiments}\label{subsec:regression}
We evaluate our VRCP framework on regression tasks from the PettingZoo
Multi-Particle Environment (MPE) library \cite{terry2021pettingzoo}
for deep reinforcement learning. In these environments, the world is a
2D space containing \(n\) agents (of which some may be adversarial)
and \(m\) landmarks, which are defined as circles of fixed radii. The
position of the landmarks is fixed, and agents traverse the space
according to second-order motion laws. We evaluate our method on three
tasks:

\begin{itemize}
\item \textbf{Adversary} The {good} agents must try to reach a
  specific {goal} landmark whilst avoiding the adversaries. We use 2
  good agents, 1 adversary and 2 landmarks.
\item \textbf{Spread} All agents collaborate and minimise the distance
  to each landmark. We set the number of agents and landmarks equal to
  3.
\item \textbf{Push} In this task, there is a single good agent,
  adversary and landmark. The task is for the adversary to maximise
  the distance between the landmark and the good agent.
\end{itemize}

In our experiments, for data-generation we select 5,000 random initial
world configurations and, for each, simulate 25 Monte-Carlo
trajectories of length \(k=5\). The regression task for all
environments is to predict the upper and lower quantiles of the total
cumulative reward over the \(k\) steps, given as input the initial
world state. As in the classification experiments, we partition the
dataset into the following partitions:
\(|\sD_\text{train}| = 1{,}000\), \(|\sD_\text{cal}| = 2{,}000\) and
\(|\sD_\text{test}| = 2{,}000\).

For computing the bounds, we use CROWN \cite{zhang2018efficient} with
\(\ell_{\infty}\)-bounded perturbations. To generate the adversarially
perturbed test points, we use the Fast Gradient Sign Method as given
in \citep{goodfellow2015explaining}.

\begin{table}
  \centering
  \caption{Marginal coverage and average interval lengths for each MPE
    regression task for various \(\epsilon\) perturbations bounded by
    an \(\ell_\infty\)-norm. All results record a 95\% confidence
    interval with \(n_{\text{splits}} = 50\).}
  \label{table:mpe_results}
  \resizebox{\columnwidth}{!}{%
    \begin{tabular}{c c c c c c c c }
        \toprule
        & Perturbation & \multicolumn{2}{c}{\(\epsilon = 0.01\)} & \multicolumn{2}{c}{\(\epsilon = 0.02\)} & \multicolumn{2}{c}{\(\epsilon = 0.04\)} \\
      \cmidrule(rl){2-2}
      \cmidrule(rl){3-4}
      \cmidrule(rl){5-6}
      \cmidrule(rl){7-8}
        & Method & Coverage & Length & Coverage & Length & Coverage & Length \\
      \toprule
      \multirow{3}{*}{\rotatebox[origin=c]{90}{{\scriptsize \textbf{Adversary}}}}
        & Vanilla & {\color[HTML]{CC0000} 0.871±0.006} & {\color[HTML]{CC0000} 0.480±0.006} & {\color[HTML]{CC0000} 0.834±0.007} & {\color[HTML]{CC0000} 0.484±0.006} & {\color[HTML]{CC0000} 0.745±0.009} & {\color[HTML]{CC0000} 0.490±0.006}\\
      \cmidrule(lr){2-8}
        & VRCP--I & 0.928±0.004 & 0.605±0.006 & 0.951±0.003 & 0.673±0.006 & 0.985±0.002 & 0.855±0.006 \\
        & VRCP--C & 0.910±0.005 & 0.534±0.006 & 0.923±0.005 & 0.606±0.006 & 0.966±0.003 & 0.806±0.005 \\
      \toprule
      \multirow{3}{*}{\rotatebox[origin=c]{90}{{\scriptsize \textbf{Spread}}}} 
        & Vanilla & {\color[HTML]{CC0000} 0.864±0.005} & {\color[HTML]{CC0000} 0.595±0.005} & {\color[HTML]{CC0000} 0.834±0.005} & {\color[HTML]{CC0000} 0.602±0.005} & {\color[HTML]{CC0000} 0.768±0.006} & {\color[HTML]{CC0000} 0.612±0.005}\\
        \cmidrule(lr){2-8}
        & VRCP--I & 0.929±0.004 & 0.690±0.006 & 0.958±0.003 & 0.769±0.006 & 0.991±0.001 & 0.992±0.006 \\
        & VRCP--C & 0.908±0.005 & 0.663±0.006 & 0.935±0.004 & 0.762±0.005 & 0.977±0.002 & 1.054±0.006 \\
      \toprule
      \multirow{3}{*}{\rotatebox[origin=c]{90}{{\scriptsize \textbf{Push}}}} 
        & Vanilla & {\color[HTML]{CC0000} 0.891±0.006} & {\color[HTML]{CC0000} 0.643±0.006} & {\color[HTML]{CC0000} 0.875±0.007} & {\color[HTML]{CC0000} 0.646±0.006} & {\color[HTML]{CC0000} 0.841±0.008} & {\color[HTML]{CC0000} 0.652±0.006}\\
      \cmidrule(lr){2-8}
        & VRCP--I & 0.917±0.006 & 0.687±0.006 & 0.934±0.005 & 0.721±0.006 & 0.961±0.003 & 0.800±0.006 \\
        & VRCP--C & 0.905±0.005 & 0.674±0.006 & 0.910±0.005 & 0.711±0.005 & 0.924±0.005 & 0.795±0.005 \\
      \bottomrule
    \end{tabular}
  }
\end{table}

As seen in \cref{table:mpe_results}, both VRCP methods recover the
marginal coverage guarantees in the presence of adversarial
perturbations, whereas vanilla CP fails drastically after
\(\epsilon = 0.02\). We note that the performance of VRCP--C and
VRCP--I are similar, although VRCP--I tends to produce more
conservative intervals (without sacrificing efficiency).


%
\section{Limitations}
\label{sec:limitations}
VRCP's scalability depends on that of the underlying neural network
verifier. We evaluated VRCP on small to medium-sized neural networks.
For large networks, existing complete verification methods would
become computationally infeasible, while incomplete methods would
produce bounds that are too loose to be useful.  However, it is
important to note that since VRCP is agnostic of the specific
verification tool used, it would directly benefit from any future
advances in neural network verification. Thus, as neural network
verification tools continue to evolve and improve, so does VRCP.


%
\section{Conclusion}
We introduced Verifiably Robust Conformal Prediction (VRCP), a novel
framework that leverages conformal prediction and neural network
verification to produce prediction sets that maintain marginal
coverage under adversarial perturbations. We presented two variants:
VRCP--C, which applies verification at calibration time, and VRCP--I,
which applies verification at inference time.

Extensive experiments on classification and regression tasks
demonstrated that VRCP recovers valid marginal coverage in the
presence of \(\ell_1\), \(\ell_2\), and \(\ell_\infty\)-norm bounded
adversarial attacks while producing more accurate prediction sets than
existing methods. VRCP is the first adversarially robust CP framework
supporting regression tasks and perturbations beyond the
\(\ell_2\)-norm, achieving strong results without relying on
probabilistic smoothing or posthoc corrections.  VRCP's theoretical
guarantees and empirical performance showcase the potential of
leveraging verification tools for uncertainty quantification of
machine learning models under attack.


%
\begin{ack}
This work is supported by the “REXASI-PRO” H-EU project, call
HORIZON-CL4-2021-HUMAN-01-01, Grant agreement ID: 101070028.


\end{ack}

\bibliography{bibliography}
\bibliographystyle{IEEEtranN}

\newpage
\appendix
\section{Additional Proof Details}
\label{sec: proofDetails}
Here we prove \cref{prop: Containment} regarding the prediction sets
obtained from VRCP--I and VRCP--C.
\begin{proof}[Proof of \cref{prop: Containment}]\label{proof:containment}
  To prove \(C(\vx) \subseteq C_{\epsilon}(\vx)\) for VRCP--I it
  suffices to observe that
  \begin{align*}
    C_{\epsilon}({\vx}) = \set{y \in Y : \lb{s}({\vx},  y) \leq Q_{1 - \alpha}(F)} 
    & \supseteq \bigcup_{\vx'\in B_{\epsilon}({\vx}_{n+1})}\set{y \in Y : S({\vx'},  y) \leq Q_{1 - \alpha}(F)}\\ 
    & = \bigcup_{\vx'\in B_{\epsilon}({\vx}_{n+1})}C(\vx') \supseteq C(\vx).
  \end{align*}
  To prove the same for VRCP--C, we observe that since
  \(Q_{1 - \alpha}(\ub{F}) \geq Q_{1 - \alpha}(F)\), we have that
  \begin{equation*}
    C_{\epsilon}({\vx}) = \set{y \in Y : S({\vx},  y) \leq Q_{1 - \alpha}(\ub{F})} \supseteq \set{y \in Y : S({\vx},  y) \leq Q_{1 - \alpha}(F)} = C(\vx). \qedhere
  \end{equation*}
\end{proof}

Next, we prove the validity of the VRCP--I region for the regression
case, defined in \cref{eq:robust_cqr_region}.

\begin{proof}
  It suffices to show that all
  $y\in C_{\epsilon}(\vx_{n+1}) = [f^\bot -q, f^\top + q]$ satisfy
  $\lb{s}(\vx_{n+1},y) = \max{\{f^\bot - y, y - f^\top \} } \leq q$
  and all $y\not\in C_{\epsilon} (\vx_{n+1})$ do not. For simplicity
  of notation, we abbreviated
  $\ub{f}_{\alpha_{\text{high}}}(\vx_{n+1})$ with $f^\top$,
  $\lb{f}_{\alpha_{\text{low}}}(\vx_{n+1})$ with $f^\bot$ and
  $Q_{1-\alpha}(F)$ with $q$.

  Assume $y\in C_{\epsilon}(\vx_{n+1})$. We divide the proof into two
  cases:
  \begin{enumerate}
  \item $\lb{s}(\vx_{n+1},y)=f^\bot - y$, which implies that
    $y \in [f^\bot -q,\frac{f^\top+f^\bot}{2}]$. It suffices to show
    that $f^\bot - y\leq q$ for $y=f^\bot -q$, which is clearly
    satisfied.
  \item $\lb{s}(\vx_{n+1},y)= y - f^\top$, which implies that
    $y \in [\frac{f^\top+f^\bot}{2},f^\top+q]$. It suffices to show
    that $y - f^\top\leq q$ for $y=f^\top +q$, which is clearly
    satisfied.
  \end{enumerate}
  Finally, we show that $y\not\in C_{\epsilon} (\vx_{n+1})$ implies
  $\lb{s}(\vx_{n+1},y)>q$: if $y < f^\bot -q$, we have that
  $\lb{s}(\vx_{n+1},y)=f^\bot - y>q$. Similarly, if $y > f^\bot +q$,
  we have that $\lb{s}(\vx_{n+1},y)=y-f^\top>q$.
\end{proof}
\section{Model Details}
\label{sec: modelDetails}
All experimental results were obtained from running the code provided
in our GitHub repository on a server with 2x Intel Xeon Platinum 8360Y
(36 cores, 72 threads, 2.4GHz), 512GB of RAM and an NVIDIA A40 48GB
GPU. All pre-trained models as well as the training scripts are also
provided in the GitHub repository. In summary, the models' train and
test performances are provided in \cref{tbl: Classification
  Models,tbl: Regression Models}.

\begin{table}[hb]
  \centering
  \caption{Train and test accuracies (\%) for the classifications
    models on CIFAR10, CIFAR100, and TinyImageNet datasets. It should
    be noted that the model's accuracy has no effect on VRCP’s
    validity and only affects the efficiency of the prediction sets
    (more accurate models, tighter prediction regions)}
  \label{tbl: Classification Models}
  \begin{tabular}{lccc}
    Metric      & CIFAR10 & CIFAR100 & TinyImageNet \\
    \midrule
    Train Top-5 & 98.77   & 90.49    & 78.44 \\
    Train Top-1 & 77.80   & 67.12    & 52.81 \\
    Test Top-5  & 98.27   & 82.87    & 55.72 \\
    Test Top-1  & 76.52   & 55.73    & 29.65
  \end{tabular}
\end{table}
\subsection{Classification}

\paragraph{CIFAR10} We use 2 convolution layers with average pooling
and dropout, followed by 2 linear layers. ReLU activations across all
layers.

\paragraph{CIFAR100} We use 1 convolution layer with average pooling,
2 further convolution layers with average pooling and dropout followed
by 2 linear layers. ReLU activations across all layers.

\paragraph{TinyImageNet} We use 4 convolution layers with dropout
followed by 2 linear layers with dropout. Leaky ReLU activation
function with \(a=0.1\)

For all models we train using images augmented with random crop with 4
pixels of padding and random horizontal flip. We standardise the
TinyImageNet models with \(\mu=0.5\) and \(\sigma=0.5\) overall 3 RGB
channels.

As previously mentioned, we do not make any assumptions during
training about the perturbations we expect to see at inference
time. As such, unlike the existing SotA methods, we do not train on
smoothed or adversarially attacked images.

All models are trained for 200 epochs with a batch size of 128 using
the stochastic gradient descent optimiser with momentum set to 0.9. We
also employ a weight decay of \(5\times10^{-4}\) and a cosine
annealing learning rate scheduler.

\subsection{Regression}
For the MPE datasets, we train Deep Q-Net policies for the RL tasks
for the sole purposes of generating the appropriate datasets and
provide these policies in the GitHub repository.

The model used for the quantile regressors is a simple linear
architecture consisting of 3 layers, separated with ReLU activation
functions and dropout. We trained the model to estimate the $\alpha/2$
and $1-\alpha/2$ quantiles, where $\alpha = 0.1$, as in the other
experiments.

The exact parameters for the RL policies can be found in the config
files within the GitHub repository, however have little bearing on the
efficiency of our results, being used only for the data-generating
process. The quantile regressors are each trained for 400 epochs, with
a learning rate of \(10^{-5}\), dropout of 0.1 and a decay of
\(10^{-5}\).

\begin{table}[ht]
  \centering
  \caption{Train and test loss for the regression models in the
    adversary, spread, and push environments.}
  \label{tbl: Regression Models}
  \begin{tabular}{lccc}
    Metric & Adversary & Spread & Push \\
    \midrule
    Train  & 0.066     & 0.075 & 0.075 \\
    Test   & 0.051     & 0.053 & 0.068
  \end{tabular}
\end{table}
%


\end{document}